\documentclass[12pt]{amsart}

\usepackage[left=15mm, right=15mm, top=10mm, bottom=15mm]{geometry}

\usepackage[english]{babel}
\usepackage[T1]{fontenc}
\usepackage[utf8]{inputenc}

\theoremstyle{plain}
    \newtheorem{theorem}{Theorem}[section]
    \newtheorem{lemma}[theorem]{Lemma}
    \newtheorem{proposition}[theorem]{Proposition}

\theoremstyle{definition}

\usepackage[
    draft = false,
    unicode = true,
    colorlinks = true,
    allcolors = blue,
    hyperfootnotes = true,
    citecolor = red
]{hyperref}
\usepackage{amsfonts,amsmath,amssymb,amscd,amsthm,url}
\usepackage[inline]{enumitem}
\usepackage{graphicx,epsf,afterpage, wrapfig}
\usepackage{soul}
\usepackage{multicol}
\usepackage{array}
\usepackage{braket}
\usepackage{epigraph}
\usepackage{rotating, floatflt}
\usepackage{xstring} 

\usepackage[
backend=biber,
style=numeric-comp,
sorting=none,
giveninits=true
]{biblatex}
\usepackage{csquotes}

\usepackage{xcolor}

\usepackage{ulem}

\usepackage{tikz}
\usetikzlibrary{positioning, calc, intersections, through, arrows, matrix, chains, math}
\usetikzlibrary{arrows.meta}
\usetikzlibrary{shadows}
\usetikzlibrary{graphs}
\usetikzlibrary{graphs.standard}
\usetikzlibrary{patterns}
\usetikzlibrary{decorations.pathreplacing,calligraphy,backgrounds}

\DeclareFieldFormat{usera}{\href{https://arxiv.org/abs/#1}{\texttt{arXiv:#1}}}

\renewbibmacro*{finentry}{\printfield{usera}\newunit\finentry}

\renewbibmacro{in:}{%
  \iffieldundef{usera}
    {\printtext{\bibstring{in}\intitlepunct}}
    {}%
}

\newcommand\norm[1]{\ensuremath{\left\lVert#1\right\rVert}}
\newcommand\abs[1]{\ensuremath{\left\lvert#1\right\rvert}}

\newcommand{\Bcal}{\mathcal{B}}

\newcommand{\Dcal}{\mathcal{D}}

\newcommand{\Hcal}{\mathcal{H}}

\newcommand{\Pcal}{\mathcal{P}}

\newcommand{\Ucal}{\mathcal{U}}

\newcommand{\Dbf}{{\ensuremath{\mathbf{D}}}}

\newcommand{\Pbf}{{\ensuremath{\mathbf{P}}}}

\newcommand{\Ubf}{{\ensuremath{\mathbf{U}}}}

\newcommand{\Phbf}{{\ensuremath{\mathbf{\Phi}}}}

\newcommand{\defing}[1]{\textbf{\emph{\mathversion{bold}#1}}}

\newcommand{\Z}{\ensuremath{\mathbb{Z}}}
\newcommand{\Cx}{\ensuremath{\mathbb{C}}}

\newcommand{\N}{\ensuremath{\mathbb{N}}}

\renewcommand{\geq}{\geqslant}

\newcounter{mcnt}

\newcounter{wordcnt}

\begin{document}

\title{Quantum channels preserving sigma-additivity \\ on diagonal von Neumann algebras \\ and Ulam measurable cardinals}

\author{Sviatoslav.\,V. Dzhenzher}

\begin{abstract}
     This paper investigates the interplay between the properties of quantum states on the Hilbert space \(\ell_2(\kappa)\) and the set-theoretic nature of the cardinal $\kappa$.
     We focus on the existence of singular $\sigma$-additive states~--- functionals whose induced measures are $\sigma$-additive yet vanish on singletons.
     While the existence of such states is known to be equivalent to the Ulam measurability of $\kappa$, their structural and dynamical properties remain largely unexplored.
     We prove that any $\sigma$-additive state on the diagonal algebra is representable as a Pettis integral over a singular $\sigma$-additive measure, extending the classical representation theory to the non-normal sector.
     Furthermore, we construct a class of quantum channels using $\sigma$-complete ultrafilters that map normal states to singular $\sigma$-additive states, effectively <<archiving>> information into the singular part of the state space.
\end{abstract}

\thanks{\hspace{-4mm}
S.\,V. Dzhenzher: sdjenjer@yandex.ru. orcid: 0009-0008-3513-4312
\\
Moscow Institute of Physics and Technology 141701, Institutskii per. 9, Dolgoprudny, Russia}

\maketitle
\thispagestyle{empty}

\emph{Keywords: measurable cardinals; quantum dynamics; singular quantum states; Pettis integrals.}

\vspace{5mm}

\emph{MSC: 03E55, 46L30; 81P47, 46G10}

\section{Introduction}

The mathematical foundations of quantum mechanics rely on the representation of physical states as linear functionals on an algebra of observables. For a Hilbert space $\Hcal$, the state space \(\Sigma(\Hcal)\) is traditionally decomposed into normal states, which are $*$-weakly continuous and uniquely represented by trace-class operators, and singular states, which vanish on the ideal of compact operators \cite{Yosida-Hewitt}.

A powerful tool for the analysis of state structures is the Pettis integral, which allows spectral-like decomposition of states over the set of pure vector states. As established by Amosov and Sakbaev \cite{Amosov-Sakbaev-2013, Amosov-Sakbaev-2025}, in the standard separable setting, a state is normal if and only if it can be represented as a Pettis integral over a discrete $\sigma$-additive measure. This equivalence reinforces the intuition that $\sigma$-additivity is a hallmark of normality.

While a general quantum channel or operation is mathematically defined as a completely positive map on the state space \cite{Kholevo-qsys-chan-inf}, its operational applications in quantum communication typically avoid singular structures due to the strict requirements of information transmission and energy constraints \cite{Kholevo-qprobstat, Nielsen-Chuang}.
However, recent developments in the study of quantum dynamics on infinite-dimensional systems have challenged this boundary. Investigations into random shifts, operator walks, and limit dynamics on Abelian algebras have shown that singular states naturally emerge as limits of physical processes \cite{Volovich-Sakbaev-2014, Volovich-Sakbaev-2018, Orlov-Sakbaev-25, DzhenzherDzhenzherSakbaev26}. This raises two fundamental questions:
\begin{enumerate}[label=\arabic*.]
    \item can a state be simultaneously singular and $\sigma$-additive?

    \item can a normal state be Pettis representable over a $\sigma$-additive measure which vanishes on singletons?
\end{enumerate}
In standard ZFC set theory, if there exists a completely-additive measure that vanishes on singletons on the powerset of some set, then this set has to be a measurable cardinal \cite{Jech-p1-ch10}. Thus, the existence of non-normal $\sigma$-additive states is intrinsically linked to the existence of measurable cardinals.

The existence of such states was recently confirmed by Blecher and Weaver \cite{Blecher-Weaver}, who proved that singular $\sigma$-additive states exist on \(\Bcal(\ell_2(\kappa))\) if and only if $\kappa$ is an Ulam measurable cardinal. While their work provides a definitive set-theoretic classification, the operational nature and the internal structure of these states remain to be explored.
In this paper, we bridge the gap between the axiomatic results of \cite{Blecher-Weaver} and the representation theory developed in \cite{Amosov-Sakbaev-2013, Amosov-Sakbaev-2025}. Our contribution is twofold:

\begin{enumerate}[label=\arabic*.]
    \item We demonstrate that the correspondence between states and Pettis integrals <<survives>> the transition to the singular sector on measurable cardinals. Specifically, we prove that any $\sigma$-additive state on the diagonal algebra \(\Dcal(\ell_2(\kappa))\cong\ell_\infty(\kappa)\) is a Pettis integral over its induced $\sigma$-additive measure, even when that measure vanishes on singletons.

    \item We construct an explicit class of quantum channels using $\sigma$-complete ultrafilters. We show that these channels act as a mechanism for <<information archiving>>, mapping normal states into singular $\sigma$-additive outcomes. This provides a dynamical interpretation of measurable cardinals as a domain where quantum information can be preserved in a $\sigma$-additive manner while becoming inaccessible to finite-dimensional observations.
\end{enumerate}

The paper is organized as follows: Section~\ref{s:mc} provides the necessary background on large cardinals; Section~\ref{s:sing} establishes the Pettis representation of singular states; and Section~\ref{s:qc} introduces and analyzes the properties of the archiving quantum channels.

\section{Measurable cardinals}\label{s:mc}

Let \(\{r_i:i\in I\}\) be a collection of non-negative real numbers.
Recall that
\[
    \sum_{i \in I}r_i = \sup \left\{\sum_{i \in E} r_i : \text{$E\subset I$ is finite}\right\}.
\]
If the sum is finite, then there is at most countable countable collection of non-zero $r_i$.

For a set $X$, denote by \(\Pcal(X)\) the set of subsets of $X$.

Let \(\kappa>\aleph_0\) be a cardinal.
Recall that a measure \(\mu\colon \Pcal(\kappa)\to [\,0,1\,]\) is a finitely-additive map.
In this paper we always assume probability measures; that is, \(\mu(\kappa)=1\).
Following \cite{Blecher-Weaver}, we say that a measure \(\mu\colon \Pcal(\kappa)\to [\,0,1\,]\) is \defing{$<\!\!\kappa$-additive}, if for any collection \(\{A_i:i\in I\}\) of pairwise disjoint subsets \(A_i \subset \kappa\) with \(\abs{I} < \kappa\) we have
\[
    \mu\left(\bigsqcup_{i \in I} A_i\right) = \sum_{i\in I} \mu(A_i).
\]
Note that the usual definition of \emph{$\kappa$-additivity} is not consistent with $\sigma$-additivity, when exactly countable collection is taken.
That is why we use the notation $<\!\!\kappa$.

An ultrafilter \(\Ucal\subsetneq \Pcal(\kappa)\) on a cardinal \(\kappa\) is said to be \defing{$<\!\!\kappa$-complete}, if for any collection \(\{A_i:i\in I\}\) of elements \(A_i \in \Ucal\) with \(\abs{I} < \kappa\), their intersection also lies in the ultrafilter:
\[
    \bigcap_{i\in I} A_i \in \Ucal.
\]
Note that if an ultrafilter is $<\!\!\kappa$-complete for \(\kappa>\aleph_0\), then it is $\sigma$-complete, meaning that it is closed under countable intersections.
An ultrafilter is \defing{non-principal} if it does not have a least element; in other words, if the intersection of its elements is empty.

There is a strong connection between two-valued measures and ultrafilters.
So, given an ultrafilter \(\Ucal\subsetneq \Pcal(\kappa)\), one may define the two-valued measure \(\mu_\Ucal\colon \Pcal(\kappa)\to \{0,1\}\) by
\[
    \mu_\Ucal(A) := \begin{cases}
        1, &\text{if $A\in \Ucal$}, \\
        0, &\text{else}.
    \end{cases}
\]
Conversely, given a two-valued probability measure \(\mu\), one may construct an ultrafilter of sets of measure~$1$.
It is easy to verify that
\begin{itemize}
    \item a measure \(\mu_\Ucal\) is ($\sigma$-) $<\!\!\kappa$-additive if and only if \(\Ucal\) is ($\sigma$-) $<\!\!\kappa$-complete, and
    \item a measure \(\mu_\Ucal\) vanishes on singletons if and only if \(\Ucal\) is non-principal.
\end{itemize}

Following \cite{Blecher-Weaver}, we say that a cardinal $\kappa>\aleph_0$ is

\begin{itemize}
    \item \defing{measurable}, if there exists a $<\!\!\kappa$-additive two-valued measure \(\Pcal(\kappa)\to\{0,1\}\) that vanishes on singletons;

    \item \defing{real-valued measurable} if there exists a $<\!\!\kappa$-additive probability measure \(\Pcal(\kappa)\to [\,0,1\,]\) that vanishes on singletons;

    \item \defing{Ulam measurable}, if there exists a $\sigma$-additive two-valued measure \(\Pcal(\kappa)\to\{0,1\}\) that vanishes on singletons;

    \item \defing{Ulam real-valued measurable} if there exists a $\sigma$-additive probability measure \(\Pcal(\kappa)\to [\,0,1\,]\) that vanishes on singletons.
\end{itemize}

\begin{figure}
    \centering
\begin{tikzpicture}[
    >=Stealth,
    node distance=2cm,
    cardinal/.style={rectangle, draw=blue!60, fill=blue!5, very thick, minimum size=10mm, font=\sffamily\bfseries},
    descr/.style={font=\small\itshape, text width=4cm, align=center}
]

    \node[cardinal] (MC) {MC};

    \node[cardinal, below=1.5cm of MC] (RVMC) {RVMC};

    \node[cardinal, right=2cm of RVMC] (URVMC) {URVMC};

    \node[cardinal, right=2.6cm of MC] (UMC) {UMC};

    \draw[->, thick] (MC) -- (RVMC) {};
    
    \draw[->, thick] (RVMC) -- (URVMC);
        
    \draw[->, thick] (UMC) -- (URVMC);

    \draw[->, thick] (MC) -- (UMC);

\end{tikzpicture}
\caption{Obvious relations between measurability of caridinals. MC is measurable cardinal. UMC is Ulam measurable cardinal. RVMC is real-valued measurable cardinal. URVMC is Ulam real-valued measurable cardinal}
    \label{fig:mc-relations}
\end{figure}
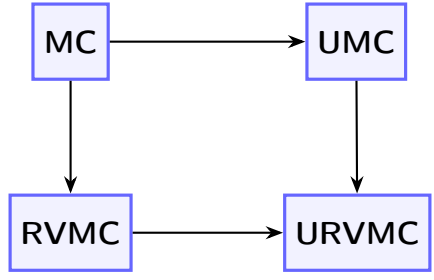

Note that (Ulam) measurability is equivalent to the following: $\kappa$ is (Ulam) measurable, if there exists a non-principal ($\sigma$-) $<\!\!\kappa$-complete ultrafilter.
See some trivial relations between these definitions in Figure~\ref{fig:mc-relations}.
It is also known that a real-valued measurable cardinal $\kappa$ is measurable if \(\kappa > \mathfrak{c}\), and no measurable cardinal can be less than $\mathfrak{c}$.
Moreover, any cardinal $\kappa$ is Ulam measurable if and only if it is equal to or greater than some measurable cardinal, which is inaccessible, and so \(\kappa > \mathfrak{c}\).
Finally, any cardinal $\kappa$ is Ulam real-valued measurable if and only if it is equal to or greater than some real-valued measurable cardinal, which is weakly inaccessible; in this case, if \(\mathfrak{c} \geq \kappa\), the Continuum Hypothesis is false.
See, for example, \cite{Jech-p1-ch10, Blecher-Weaver} for details.

\section{Pettis representability of singular sigma-additive states}\label{s:sing}

Fix any (not necessarily measurable) cardinal \(\kappa>\aleph_0\).
We will consider the Hilbert space \(\Hcal := \ell_2(\kappa)\) of functions \(f\colon \kappa\to\Cx\) with at most countable domain and the norm
\[
    \norm{f} := \sqrt{\sum_{i < \kappa}\abs{f(i)}^2}.
\]
Denote by \(\{e_i\}\) its orthonormal basis.

Denote by \(\Bcal(\Hcal)\) the Banach algebra of linear bounded operators \(\Hcal\to\Hcal\).
We restrict our attention to the Banach subalgebra \(\Dcal(\Hcal) \subset \Bcal(\Hcal)\) of diagonal operators.
It is clear that \(\Dcal(\Hcal)\) is isomorphic to \(\ell_\infty(\kappa)\).
The elements of \(\Dcal(\Hcal)\) are called \defing{observables}.
Note that the projectors from \(\Dcal(\Hcal)\) are exactly the diagonal operators with eigenvalues equal to $1$.
Thus, for \(I\subset\kappa\) we denote by \(\Pbf_I\) the projector onto the subspace generated by basis vectors \(\{e_i\}_{i\in I}\).
In particular, \(\Pbf_\kappa\) is the identity operator.

Denote by \(\Sigma(\Hcal) := S^{1+}(\Dcal(\Hcal)^*)\) the space of \defing{quantum states}, that is, the linear functionals \(\Dcal(\Hcal) \to \Cx\), lying on the intersection of the positive cone and the unit sphere.
We denote the \defing{action} of a quantum state \(\rho \in \Sigma(\Hcal)\) on an observable \(\Dbf\in \Dcal(\Hcal)\) by
\[
    \braket{\rho, \Dbf}.
\]

In quantum mechanics usually consider the subspace \(\Sigma_n(\Hcal) \subset \Sigma(\Hcal)\) of \defing{normal quantum states}, which are $*$-weakly continuous quantum states.
It is known \cite{Blecher-Weaver} that a state \(\rho\in\Sigma(\Hcal)\) is normal if and only if it is \defing{completely additive}, which means that the equality
\[
    \Braket{\rho, \sum_\alpha \Pbf_\alpha} = \sum_\alpha \Braket{\rho, \Pbf_\alpha}
\]
holds for any family of pairwise orthogonal projectors \(\Pbf_\alpha\).
For separable Hilbert spaces, it is known that normality is equivalent to the fact that the supremum of the action over projectors onto finite-dimensional subspaces equals one \cite{Volovich-Sakbaev-2018}. In our case, this is also true.

\begin{proposition}
    A state \(\rho\) is normal if and only if
    \[
        \sup_{\text{$I\subset\kappa$ is finite}} \Braket{\rho, \Pbf_I} = 1,
    \]
    where the supremum is taken over projectors onto finite-dimensional subspaces.
\end{proposition}

\begin{proof}
    Suppose \(\rho\) is normal. Then it is completely additive. Hence,
    \[
        1 = \Braket{\rho, \Pbf_\kappa} = \sum_{i<\kappa} \Braket{\rho, \Pbf_i} = \sup_{I\subset\kappa, \text{ finite}} \Braket{\rho, \Pbf_I},
    \]
    where the first equality follows from the normalization of states, the second equality follows from complete additivity, and the last equality follows from the definition of the sum.

    Conversely, suppose that
    \[
        \sup_{I\subset\kappa, \text{ finite}} \Braket{\rho, \Pbf_I} = 1.
    \]
    By the definition of the supremum, for any \(I\subset\kappa\), we have
    \[
        \Braket{\rho, \Pbf_I} \geq \sum_{i\in I}\Braket{\rho, \Pbf_i}.
    \]
    Then
    \[
        1 = \Braket{\rho, \Pbf_\kappa} = \Braket{\rho, \Pbf_I} + \Braket{\rho, \Pbf_{\kappa\setminus I}} \geq \sum_{i \in I} \Braket{\rho, \Pbf_i} + \sum_{i \in \kappa\setminus I} \Braket{\rho, \Pbf_i} = \sum_{i<\kappa} \Braket{\rho, \Pbf_i} = 1.
    \]
    Thus, the inequalities become equalities, and for any $I\subset \kappa$, we have
    \[
        \Braket{\rho, \Pbf_I} = \sum_{i\in I} \Braket{\rho, \Pbf_i}.
    \]
    Then for an arbitrary family \(\{I_\alpha\}\) such that \(I = \bigsqcup_\alpha I_\alpha\), we have
    \[
        \Braket{\rho, \Pbf_I} = \sum_{i\in I} \Braket{\rho, \Pbf_i} = \sum_{\alpha}\sum_{i\in I_\alpha} \Braket{\rho, \Pbf_i} = \sum_\alpha \Braket{\rho, \Pbf_{I_\alpha}},
    \]
    where the rearrangement of terms is possible due to the absolute convergence of the series.
    The equality of the first and last expressions is precisely complete additivity.
\end{proof}

For \(u\in S^1(\Hcal)\), a normal quantum state \(\rho=\rho_u\in\Sigma_n(\Hcal)\) is called a \defing{pure vector state}, if for all \(\Dbf\in\Dcal(\Hcal)\) we have
\[
    \Braket{\rho, \Dbf} = (u, \Dbf u).
\]

We will be more interested in another quantum states.
So, we say that a state \(\rho\) lies in the subspace \(\Sigma_s(\Hcal)\subset\Sigma(\Hcal)\) of \defing{singular quantum states} if it vanishes on all rank $1$ projectors; or equivalently, on all compact operators.
In other words, if
\[
    \Braket{\rho, \Pbf_I} =0
\]
for any finite \(I\subset\kappa\).
The Yosida--Hewitt decomposition \cite{Yosida-Hewitt} claims that
\[
    \Sigma(\Hcal) = \Sigma_n(\Hcal) \oplus \Sigma_s(\Hcal),
\]
meaning that any quantum state \(\rho\in\Sigma(\Hcal)\) can be represented as a convex combination
\[
    \rho = p\rho_n + (1-p)\rho_s
\]
for some \(p\in [\,0,1\,]\), \(\rho_n\in\Sigma_n(\Hcal)\), and \(\rho_s\in\Sigma_s(\Hcal)\).

It is known, see for example \cite{Amosov-Sakbaev-2025}, that any quantum state can be represented as the Pettis integral over some finitely-additive non-negative finite measure on the unit sphere of pure vector states.
Since we simplified the algebra of observables to just diagonal operators, we may consider the particular case of such measures, concentrated on the pure vector states on basis vectors \(e_i\).
For a state \(\rho \in \Sigma(\Hcal)\), define the (finitely-additive) \emph{induced measure} \(\nu_\rho\colon \Pcal(\kappa)\to[\,0,1\,]\) by
\[
    \nu_\rho(I) := \Braket{\rho, \Pbf_I}.
\]

\begin{lemma}
    Any state \(\rho\in\Sigma(\Hcal)\) can be represented as the Pettis integral
    \[
        \rho = \int_{i < \kappa} \rho_{e_i}\,d\nu_\rho(i).
    \]
\end{lemma}

\begin{proof}
    We need to show that for any \(\Dbf\in\Dcal(\Hcal)\) we have
    \[
        \braket{\rho, \Dbf} = \int_{i < \kappa} (e_i, \Dbf e_i)\,d\nu_\rho(i).
    \]
    First, suppose that \(\Dbf=\Pbf_I\) for some \(I\subset\kappa\). Then
    \[
        \braket{\rho, \Dbf} = \nu_\rho(I) = \int_{i\in I}d\nu_\rho(i) = \int_{i < \kappa} (e_i, \Dbf e_i)\,d\nu_\rho(i).
    \]
    Thus, by linearity, the lemma holds for finite linear combinations
    \[
        \Dbf = \sum_{k=1}^K d_k \Pbf_{I_k}.
    \]

    Now consider the arbitrary \(\Dbf\in\Dcal(\Hcal)\).
    There is a sequence of \(\Dbf_n\) with finite spectrum approximating \(\Dbf\) by norm.
    By continuity of \(\rho\),
    \[
        \braket{\rho, \Dbf} = \lim_{n\to\infty} \braket{\rho, \Dbf_n}.
    \]
    And by the definition of the integral over finitely-additive measures, see for example \cite{Dunford-Schwartz-vol1}, 
    \[
        \lim_{n\to\infty} \braket{\rho, \Dbf_n} =
        \lim_{n\to\infty} \int_{i < \kappa} (e_i, \Dbf_n e_i)\,d\nu_\rho(i)=
        \int_{i < \kappa} \lim_{n\to\infty}(e_i, \Dbf_n e_i)\,d\nu_\rho(i)=
        \int_{i < \kappa} (e_i, \Dbf e_i)\,d\nu_\rho(i),
    \]
    which concludes the proof.
\end{proof}

Thus, there is a correspondence between quantum states and measures on subsets of \(\kappa\).
The following theorem answers on one of our questions in the introduction: a normal state cannot be represented by a measure that vanishes on singletons.

\begin{theorem}
    If a state \(\rho\in\Sigma(\Hcal)\) can be represented as the Pettis integral
    \[
        \rho = \int_{i < \kappa} \rho_{e_i}\,dm(i),
    \]
    and the measure $m$ vanishes on singletons, then $\rho$ is singular.
    
    Thus, a state \(\rho\in\Sigma(\Hcal)\) is singular if and only if the induced measure \(\nu_\rho\) vanishes on singletons.
\end{theorem}

\begin{proof}
    The first part follows since for a finite subset \(I\subset\kappa\)
    \[
        \Braket{\rho, \Pbf_I} = m(I) = 0.
    \]
    The second part follows since for a finite \(I\subset\kappa\)
    \[
        \Braket{\rho,\Pbf_I} = \nu_\rho(I)
    \]
    is zero either by the first part or by singularity of $\rho$.
\end{proof}



As we recall, the normality of a state is equivalent to the complete additivity of its induced measure.
Following the same logic, we say that a state \(\rho\in\Sigma(\Hcal)\) is \defing{$\sigma$-additive} if the induced measure \(\nu_\rho\) is $\sigma$-additive. From the above, it is clear that all normal states are $\sigma$-additive.
Furthermore, if the Hilbert space \(\ell_2(\kappa)\) were separable, then $\sigma$-additivity and complete additivity would be equivalent. In other words, on a separable Hilbert space, $\sigma$-additive and normal states are one and the same.

By \cite[Proposition~2.3.(i)]{Blecher-Weaver}, a $\sigma$-additive singular state \(\rho\) exists if and only if \(\kappa\) is an Ulam real-valued measurable cardinal.
Thus by the Ulam's result \cite[Theorem~10.1]{Jech-p1-ch10} (see the final argument of the proof of this theorem there), if there exists a $\sigma$-additive singular state \(\rho\), then the following two options are possible:

\begin{itemize}
    \item either \(\kappa\) is equal to or greater than the least measurable cardinal, which is inaccessible, and so \(\kappa > \mathfrak{c}\);

    \item or \(\mathfrak{c} \geq \kappa\), and \(\kappa\) is equal to or greater than the least real-valued measurable cardinal, which is weakly inaccessible. In this case, the Continuum Hypothesis is false.
\end{itemize}

In order to leave the Continuum Hypothesis in peace, let us now focus on the world where an Ulam measurable cardinal exists.
Let \(\kappa\) be such a cardinal.
By \cite[Proposition~2.3.(ii)]{Blecher-Weaver}, there should exist singular $\sigma$-additive pure state.
Let \(\Ucal \subsetneq \Pcal(\kappa)\) be a \(\sigma\)-complete non-principal ultrafilter.
Define the state \(\rho_\Ucal\in\Sigma(\Hcal)\) by
\[
    \Braket{\rho_\Ucal, \Dbf} := \lim_\Ucal (e_i, \Dbf e_i).
\]
Then
\[
    \nu_{\rho_\Ucal}(I) = \Braket{\rho_\Ucal, \Pbf_I} = \lim_\Ucal (e_i, \Pbf_I e_i) = \mu_\Ucal(I),
\]
and so the induced measure is $\sigma$-additive and vanishes on singletons.
Thus, given the Ulam measurable cardinal \(\kappa\), we immediately obtain the $\sigma$-additive singular pure state \(\rho_\Ucal\).
The \defing{purity} means that $\rho_\Ucal$ is an extremal point of the set \(\Sigma(\Hcal)\).

\section{Quantum channels}\label{s:qc}

Recall that usually \cite{Kholevo-qsys-chan-inf} a quantum channel is a linear trace-preserving completely positive map. In our case, since we work not only with trace-class operators, we omit the trace-preserving property and have only the preserving normalization property; see, for example, how this is done in \cite{DzhenzherDzhenzherSakbaev26}.
So, a \defing{quantum channel} is a linear completely positive map \(\Sigma(\Hcal)\to\Sigma(\Hcal)\).

Now fix an Ulam measurable cardinal \(\kappa\) with a non-principal \(\sigma\)-complete ultrafilter $\Ucal$.
Let $\kappa$ be equipped with an Abelian group structure.
For example, it can be done by taking the group \(\bigoplus_{i<\kappa} \Z_2\) of finite (that is, whose pre-image of $1$ is finite) functions \(\kappa\to\Z_2\), which is equinumerous to \(\kappa\).
Denote by \(\Ubf_j \in \Bcal(\Hcal)\) the shift operator, defined by
\[
    \Ubf_j e_i = e_{i+j}.
\]
Define the quantum channel \(\Phbf_\Ucal\) by
\[
    \Braket{\Phbf_\Ucal\rho, \Dbf} := \lim_\Ucal \Braket{\rho, \Ubf^*_j\Dbf\Ubf_j}.
\]

\begin{theorem}
    For any \(\rho\in\Sigma(\Hcal)\), the state \(\Phbf_\Ucal\rho\) is singular and preserves \(\sigma\)-additivity.
\end{theorem}

\begin{proof}
    First, let \(\rho = \rho_{e_i}\) be a pure vector state.
    Then for any \(\Dbf\in\Dcal(\Hcal)\)
    \[
        \Braket{\rho, \Ubf^*_j\Dbf\Ubf_j} = (\Ubf_j e_i, \Dbf \Ubf_j e_i) = (e_{i+j},\Dbf e_{i+j}).
    \]
    Thus\footnote{Note that, in version~1 of this paper, there was a mistake here. It was written that the outcome is \(\rho_\Ucal\), which can be so only if $\Ucal$ is translation-invariant. Apparently, there are no such ultrafilters.} \(\Phbf_\Ucal\rho = \rho_{\Ucal-i}\).
    In particular, this is a singular $\sigma$-additive state.

    Now let \(\rho\) be a normal state.
    Then it is a convex at most countable combination of vector pure states, that are, in their turn, convex at most countable combinations of vector pure states concentrated on basis vectors.
    Since \(\lim_\Ucal\) commutes with countable sums, we obtain that \(\Phbf_\Ucal\rho\) is again singular and $\sigma$-additive.

    Finally, let \(\rho\) itself be singular.
    In order to show the singularity of \(\Phbf_\Ucal\rho\), it is sufficient to prove that for any finite subset \(I \subset \kappa\) we have
    \[
        \Braket{\Phbf_\Ucal\rho, \Pbf_I} = 0.
    \]
    We have
    \[
        \Braket{\rho, \Ubf^*_j\Pbf_I\Ubf_j} = \Braket{\rho, \Pbf_{I-j}} = 0,
    \]
    and thus \(\lim_\Ucal 0 = 0\).

    Below we prove the $\sigma$-additivity of \(\Phbf_\Ucal\rho\) for a $\sigma$-additive \(\rho\).
    Denote
    \[
        \nu(I) := \Braket{\Phbf_\Ucal\rho, \Pbf_I} = \lim_\Ucal \nu_\rho(I-j).
    \]
    We need to show the $\sigma$-additivity of \(\nu\).
    It is clear that \(\nu\) is a finitely-additive measure on \(\Pcal(\kappa)\).
    Hence, it is sufficient to prove the continuity from above on \(\nu\).
    So, let \(I_1 \supset I_2 \supset \ldots\) be such that \(\bigcap_{n\in\N} I_n = \varnothing\).
    We need to prove that \(\lim_{n\to\infty} \nu(I_n)=0\).
    On the contrary, suppose that \(\lim_{n\to\infty} \nu(I_n)=L>0\).
    Then for sufficiently large $n>N$ we have
    \[
        \{j: \nu_\rho(I_n-j) > L/2 \} \in \Ucal.
    \]
    Since the ultrafilter \(\Ucal\) is $\sigma$-complete, we have
    \[
        \bigcap_{n >N}\{j: \nu_\rho(I_n-j) > L/2 \} \in \Ucal.
    \]
    In particular, this means that the intersection contains at least one element \(j_0\), common for all $n>N$.
    So, for this \(j_0\) and all (sufficiently large) $n>N$ we have \(\nu_\rho(I_n-j_0) > L/2\).
    But \(\nu_\rho(I_n-j_0) \xrightarrow[n\to\infty]{} 0\) by $\sigma$-additivity of \(\nu_\rho\). This leads to a contradiction.
\end{proof}

Thus, on Ulam measurable cardinals, it is possible to construct quantum channels that give singular $\sigma$-additive outcomes.
They destroy normal states, but preserve singularity and $\sigma$-additivity.
One may think of it as if these quantum channels archive all information somewhere far away on $\sigma$-complete ultrafilters.

\section{Conclusion}

We have demonstrated that the classification of quantum states on $\ell_2(\kappa)$ is deeply intertwined with the axiomatic foundations of set theory. Specifically, the existence of singular $\sigma$-additive states is not a property of the algebra of observables alone, but depends on whether the dimension $\kappa$ is an Ulam measurable cardinal.

Our work extends the foundational results of Blecher and Weaver \cite{Blecher-Weaver} by providing a concrete analytical and dynamical framework for these states. By establishing their representability via the Pettis integral, we have shown that the measure-theoretic structure of quantum states <<survives>> the transition to the singular sector in the presence of large cardinals. The construction of the quantum channel $\Phbf_\Ucal$ further illustrates this connection, acting as a transformative mechanism that <<shifts>> information from normal states into singular $\sigma$-additive outcomes. This result suggests that on sufficiently large dimensions, quantum information can be <<archived>> in a way that remains consistent with $\sigma$-additivity yet becomes inaccessible to any local or compact projections.

In \cite{Blecher-Weaver}, it was shown that a pure state defined on the algebra \(\Dcal(\ell_2(\kappa))\cong\ell_\infty(\kappa)\) of observables can be uniquely extended to a pure state on the whole observables algebra \(\Bcal(\ell_2(\kappa))\).
It will be interesting to connect these results with the Pettis representability of states; the problem here is that one cannot define the induced measure on the whole \(\Pcal(S^1(\ell_2(\kappa)))\) since it will not even be finitely-additive.
So it would be interesting to obtain the analogous results for the whole algebra \(\Bcal(\ell_2(\kappa))\) of observables.

\printbibliography

\end{document}